\documentclass[conference]{IEEEtran}

\usepackage{amsmath,amssymb,amsfonts}
\usepackage{amsthm}
\usepackage{algorithmic}
\usepackage{algorithm}
\usepackage{graphicx}
\usepackage{textcomp}
\usepackage{booktabs}
\usepackage{bm}

\newtheorem{proposition}{Proposition}

\begin{document}

\title{Performance Bounds and Robust Filtering for LEO Inter-Satellite Synchronization under Cross-Epoch Doppler Coupling}

\author{\IEEEauthorblockN{Haofan Dong\textsuperscript{1}, Houtianfu Wang\textsuperscript{1}, Hanlin Cai\textsuperscript{1}, and Ozgur B. Akan\textsuperscript{1,2}}
\IEEEauthorblockA{\textsuperscript{1}Internet of Everything Group, Department of Engineering, University of Cambridge, UK}
\IEEEauthorblockA{\textsuperscript{2}Center for neXt-generation Communications (CXC), Department of Electrical and Electronics Engineering,
Ko\c{c} University, T\"{u}rkiye}
}

\maketitle

% =============================================================================
% ABSTRACT
% =============================================================================
\begin{abstract}
Low Earth orbit (LEO) inter-satellite links (ISLs) must achieve
joint synchronization and ranging under severe hardware
impairments, namely oscillator phase noise, clock drift,
and measurement outliers, exacerbated by rapid relative
dynamics exceeding 7~km/s. In coherent Doppler processing,
the frequency observable depends on the \emph{difference}
between consecutive carrier phase states, creating a
cross-epoch coupling structure that fundamentally affects
estimation-theoretic performance limits. This paper makes
three contributions. First, we prove analytically that this
cross-epoch Doppler coupling is \emph{necessary} to avoid
unbounded carrier phase uncertainty: without it, phase variance grows linearly without bound. Second, we derive a posterior Cram\'{e}r-Rao
bound (PCRB) via the Tichavsk\'{y} recursion that explicitly
incorporates the resulting 10$\times$10 block information
structure. Third, we propose a hybrid robust filtering
framework combining hard gating for impulsive cycle-slip
outliers with Huber M-estimation for heavy-tail
contamination, using TASD-aware innovation covariance to
account for cross-epoch uncertainty in residual normalization.
Monte Carlo simulations at Ka-band confirm
that the PCRB accurately lower-bounds estimator performance
under nominal conditions, while the hybrid method reduces
95th-percentile phase error by 27--93\% compared to standard
extended Kalman filtering across different outlier regimes.
\end{abstract}


% =============================================================================
% SECTION I: INTRODUCTION
% =============================================================================
\section{Introduction}

Low Earth orbit (LEO) mega-constellations are reshaping global communications, positioning, navigation, and timing (PNT). Broadband LEO systems offer approximately 30~dB stronger received power and threefold improvement in satellite geometry compared to medium-Earth-orbit (MEO) GNSS~\cite{reid2018leo}, motivating proposals to fuse LEO communication links with PNT services~\cite{iannucci2022fused}. Recent experiments have demonstrated carrier-phase tracking and meter-level positioning using Starlink signals of opportunity~\cite{khalife2022starlink}. Within these constellations, inter-satellite links (ISLs) serve as the primary connectivity layer that distributes time and frequency references, making precise ISL synchronization a prerequisite for both communication capacity and navigation integrity.

Existing ISL synchronization techniques have matured primarily in the BDS-3 MEO constellation. Tang~\emph{et~al.}~\cite{tang2018isl} demonstrated centralized autonomous orbit determination using Ka-band dual one-way ranging (DOWR), achieving sub-15~cm satellite laser ranging (SLR) residuals. Xie~\emph{et~al.}~\cite{xie2020isl} constructed clock-free observables from the DOWR model, while Ruan~\emph{et~al.}~\cite{ruan2020isl} jointly estimated orbits, clocks, and ISL hardware biases from raw one-way pseudoranges. These results establish centimeter-level ranging capability for MEO ISLs. However, LEO satellites experience relative velocities exceeding 7~km/s and rapidly varying Doppler shifts, which amplify the impact of hardware impairments on synchronization performance~\cite{hauschild2021leo}.

Three coupled hardware effects complicate LEO ISL synchronization. First, oscillator phase noise, characterized by the Allan variance power-law model~\cite{allan1966statistics,zucca2005clock}, introduces both white and random-walk frequency perturbations whose accumulated effect grows with the observation interval. Second, carrier phase discontinuities (cycle slips) arising from rapid dynamics or interference~\cite{breitsch2020cycleslip} inject impulsive outliers into the measurement stream. Third, residual thermal noise under low carrier-to-noise conditions produces heavy-tail contamination. Conventional extended Kalman filter (EKF) approaches treat these impairments independently and are susceptible to filter divergence under outlier corruption. Deep learning methods~\cite{gu2024lstm} and joint synchronization-ranging formulations~\cite{gu2020sync} have been proposed but do not provide performance guarantees against theoretical limits.

% FIX #3: "renders observable" → "is necessary for bounded phase variance"
In coherent Doppler processing, the frequency-domain observable depends on the \emph{difference} between carrier phase states at consecutive epochs, i.e., $\theta_k - \theta_{k-1}$. While this differential phase structure is well-known in GNSS carrier phase processing, its formal role as a necessary condition for ISL phase observability has not been established. We term this cross-epoch coupling the \emph{time-accumulated signal difference} (TASD) structure. The TASD measurement creates a binary factor connecting adjacent state vectors, with two consequences: (i)~it is necessary to avoid unbounded carrier phase variance from Doppler measurements, and (ii)~it produces a $10 \times 10$ block Fisher information matrix (FIM) over the joint state $[\tilde{\mathbf{x}}_{k-1}^\top, \tilde{\mathbf{x}}_k^\top]^\top$ that standard single-epoch analyses neglect. While the Tichavsk\'{y} recursion~\cite{tichavsky1998pcrb} is sufficiently general to handle augmented states, existing applications to satellite positioning~\cite{closas2009crb} and Bayesian filtering~\cite{vantrees2007bayesian,simandl2001pcrb} do not exploit the explicit two-epoch FIM block structure or examine its implications for phase observability.

% FIX #6: C3 novelty emphasis on TASD-aware normalization
This paper makes three contributions:
\begin{enumerate}
\item We establish analytically that TASD coupling is \emph{necessary} to avoid unbounded phase uncertainty: with $\kappa_\theta = 0$, the posterior phase variance diverges (Proposition~\ref{prop:tasd}).
\item We derive a TASD-aware posterior Cram\'{e}r-Rao bound (PCRB) via the Tichavsk\'{y} recursion~\cite{tichavsky1998pcrb} that incorporates the $10 \times 10$ block information structure, providing a tight performance benchmark for ISL synchronization estimators.
\item We propose a hybrid robust filtering framework combining hard gating for impulsive cycle-slip outliers~\cite{sunderhauf2012switchable} with Huber M-estimation~\cite{huber1964robust,huber2009robust}, using TASD-aware innovation covariance that accounts for cross-epoch uncertainty in residual normalization. Monte Carlo simulations at Ka-band confirm 27--93\% reduction in 95th-percentile phase error compared to standard EKF across outlier regimes, with zero PCRB violation rate under nominal conditions.
\end{enumerate}

\textit{Notation:} $\mathbf{I}_n$ denotes the $n \times n$ identity matrix, $\mathbf{E}_{ij}$ the matrix with 1 at position $(i,j)$ and 0 elsewhere, and $\mathrm{blkdiag}(\cdot)$ the block-diagonal operator.

% =============================================================================
% SECTION II: SYSTEM MODEL
% =============================================================================
\section{System Model}

\begin{figure}[!t]
\centering
\includegraphics[width=0.9\columnwidth]{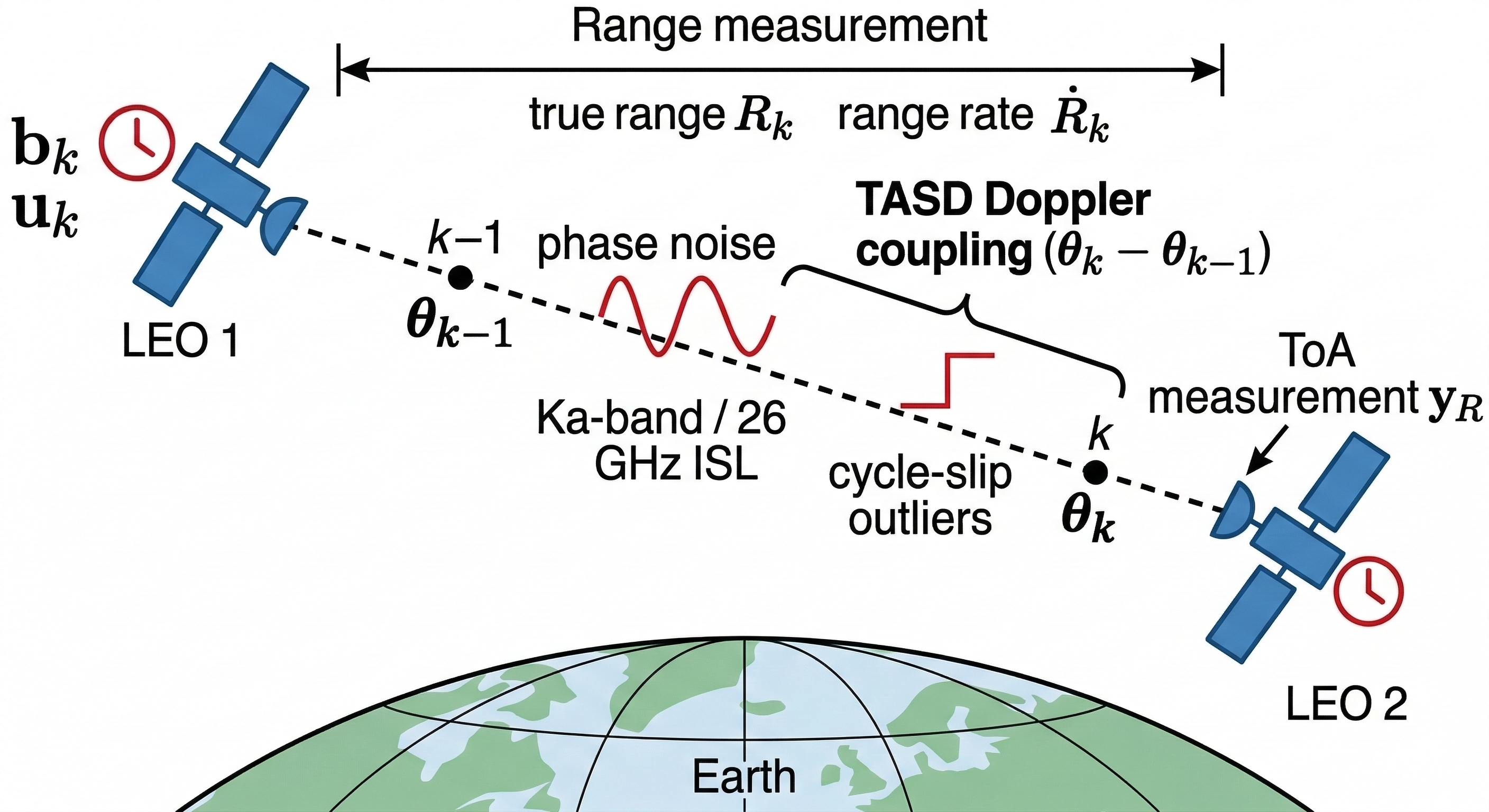}
\caption{LEO inter-satellite link synchronization scenario. Two satellites exchange ranging signals subject to clock drift, phase noise, and measurement outliers. The TASD Doppler measurement couples consecutive carrier phase states $\theta_{k-1}$ and $\theta_k$.}
\label{fig:scenario}
\end{figure}

% FIX #2: Changed "steady-state of 13 rad" to "remains within O(10) rad"
\begin{figure*}[!t]
\centering
\includegraphics[width=\textwidth]{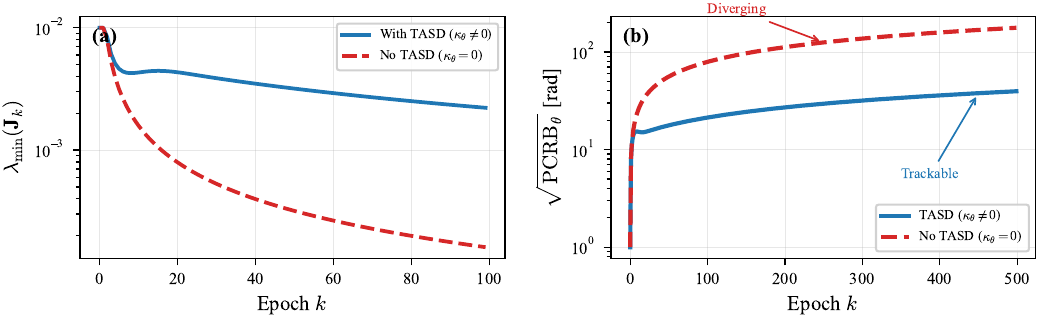}
\caption{TASD information structure and phase observability. (a)~Minimum eigenvalue of $\mathbf{J}_k$: with TASD ($\kappa_\theta \neq 0$), $\lambda_{\min}$ reaches a process-noise-limited floor; without TASD, it decays toward zero. (b)~Phase $\sqrt{[\mathbf{P}_k]_{55}}$ over 500 epochs: TASD yields sub-linear growth within $\mathcal{O}(10)$~rad, while $\kappa_\theta = 0$ diverges, confirming Proposition~\ref{prop:tasd}.}
\label{fig:info_tasd}
\end{figure*}

\subsection{State Definition and Scaling Convention}

% FIX #5: Reference Fig. 1 in text
As illustrated in Fig.~\ref{fig:scenario}, two LEO satellites exchange ranging signals over a Ka-band ISL. To maintain numerical stability and consistent units, we adopt a \emph{scaled} state vector:
\begin{equation}
\tilde{\mathbf{x}}_k = [R_k,\; \dot{R}_k,\; b_k,\; u_k,\; \theta_k]^\top
\label{eq:state}
\end{equation}
where $R_k$ [m] is the inter-satellite range, $\dot{R}_k$ [m/s] is the range rate, $b_k \triangleq c\delta_k$ [m] is the clock bias in range-equivalent meters, $u_k \triangleq c\dot{\delta}_k$ [m/s] is the clock drift in range-rate-equivalent units, and $\theta_k$ [rad] is the carrier phase. This scaling reduces the FIM dynamic range from $10^{24}$ to approximately $10^{6}$.

\subsection{State Dynamics}

The discrete-time dynamics follow a linear Gaussian model:
\begin{equation}
\tilde{\mathbf{x}}_{k+1} = \mathbf{F} \tilde{\mathbf{x}}_k + \mathbf{w}_k, \quad \mathbf{w}_k \sim \mathcal{N}(\mathbf{0}, \mathbf{Q})
\label{eq:dynamics}
\end{equation}
with state transition matrix (block-diagonal integrator structure):
\begin{equation}
\mathbf{F} = \begin{bmatrix}
1 & T & 0 & 0 & 0 \\
0 & 1 & 0 & 0 & 0 \\
0 & 0 & 1 & T & 0 \\
0 & 0 & 0 & 1 & 0 \\
0 & 0 & 0 & 0 & 1
\end{bmatrix}, \quad T = T_{\mathrm{coh}}
\label{eq:F}
\end{equation}

The process noise covariance has explicit block structure:
\begin{equation}
\mathbf{Q} = \mathrm{blkdiag}(\mathbf{Q}_R, \mathbf{Q}_b, q_\theta)
\label{eq:Q}
\end{equation}

\textbf{Range dynamics} (random acceleration model with $\sigma_a$ [m/s$^2$]):
\begin{equation}
\mathbf{Q}_R = \sigma_a^2 \begin{bmatrix} T^3/3 & T^2/2 \\ T^2/2 & T \end{bmatrix}
\label{eq:QR}
\end{equation}
The parameter $\sigma_a = 0.1$~m/s$^2$ captures residual orbital perturbations after ephemeris-based prediction.

\textbf{Clock dynamics} (Allan variance parameterization~\cite{zucca2005clock}):
\begin{equation}
\mathbf{Q}_b = c^2 \begin{bmatrix} S_f T + \frac{S_g T^3}{3} & \frac{S_g T^2}{2} \\ \frac{S_g T^2}{2} & S_g T \end{bmatrix}
\label{eq:Qb}
\end{equation}
where $S_f = h_0/2$ (white frequency noise) and $S_g = 2\pi^2 h_{-2}$ (random walk frequency noise), with $(h_0, h_{-2})$ being Allan variance power-law coefficients.

\textbf{Phase noise} (Wiener process with 3-dB linewidth $\beta$ [Hz]):
\begin{equation}
q_\theta = 2\pi \beta T \quad \text{[rad}^2\text{]}
\label{eq:qtheta}
\end{equation}

\subsection{TASD Doppler Measurement Model}

The frequency-domain Doppler observable (Hz) is:
\begin{equation}
z_D[k] = \frac{f_c}{c}\dot{R}_k + f_c\dot{\delta}_k + \frac{\theta_k - \theta_{k-1}}{2\pi T_{\mathrm{coh}}} + n_{f,k}
\label{eq:doppler_freq}
\end{equation}

Converting to range-rate equivalent units [m/s]:
\begin{equation}
y_D[k] = \dot{R}_k + u_k + \kappa_\theta(\theta_k - \theta_{k-1}) + v_{D,k}
\label{eq:tasd}
\end{equation}
where the \textbf{phase-to-range-rate coupling coefficient} is:
\begin{equation}
\kappa_\theta = \frac{c}{2\pi f_c T_{\mathrm{coh}}} \quad \text{[m/s/rad]}
\label{eq:kappa}
\end{equation}

For $f_c = 26$~GHz and $T_{\mathrm{coh}} = 0.1$~s: $\kappa_\theta = 0.0184$~m/s/rad. The dependence on $\theta_{k-1}$ creates a \emph{binary factor} connecting consecutive epochs, the TASD structure central to this work.

% FIX #7: Define ToA abbreviation
\textbf{Time-of-arrival (ToA) pseudorange measurement} [m]:
\begin{equation}
y_R[k] = R_k + b_k + v_{R,k}, \quad v_{R,k} \sim \mathcal{N}(0, \sigma_R^2)
\label{eq:toa}
\end{equation}

\subsection{TASD Jacobian Structure}

The measurement Jacobians in scaled-state coordinates are:
\begin{align}
\mathbf{H}_D &= \frac{\partial y_D}{\partial \tilde{\mathbf{x}}_k} = \begin{bmatrix} 0 & 1 & 0 & 1 & \kappa_\theta \end{bmatrix} \label{eq:HD}\\
\mathbf{H}_D^{(-)} &= \frac{\partial y_D}{\partial \tilde{\mathbf{x}}_{k-1}} = \begin{bmatrix} 0 & 0 & 0 & 0 & -\kappa_\theta \end{bmatrix} \label{eq:HDprev}\\
\mathbf{H}_R &= \frac{\partial y_R}{\partial \tilde{\mathbf{x}}_k} = \begin{bmatrix} 1 & 0 & 1 & 0 & 0 \end{bmatrix} \label{eq:HR}
\end{align}

\textbf{Physical interpretation}: The $(\dot{R}, u)$ elements have unity gain; a 1~m/s range-rate change is observationally equivalent to a 3.33~ppb clock drift in a single Doppler measurement. This ambiguity is resolved over time by their differing process noise statistics (Fig.~\ref{fig:info_tasd}(a)).

\subsection{Outlier Mechanisms}
\label{subsec:outliers}

Two measurement corruption mechanisms are used to evaluate tail reliability. Outliers corrupt Doppler measurements only; ToA remains Gaussian.

\textbf{Impulsive slips} (cycle-slip proxy): The Doppler noise is augmented by a sparse jump:
\begin{equation}
v_{D,k} = \epsilon_k \Delta_k + \tilde{v}_{D,k}, \quad \tilde{v}_{D,k} \sim \mathcal{N}(0, \sigma_D^2)
\label{eq:outlier_imp}
\end{equation}
where $\epsilon_k \sim \mathrm{Bernoulli}(p_{\mathrm{imp}})$ and $\Delta_k \sim \mathcal{N}(0, (a_{\mathrm{imp}} \sigma_D)^2)$, with $(p_{\mathrm{imp}}, a_{\mathrm{imp}}) = (0.05, 300)$.

\textbf{Heavy-tail contamination} (contaminated Gaussian): The Doppler noise follows a two-component mixture:
\begin{equation}
v_{D,k} \sim (1 - p_{\mathrm{ht}}) \mathcal{N}(0, \sigma_D^2) + p_{\mathrm{ht}} \mathcal{N}(0, (a_{\mathrm{ht}} \sigma_D)^2)
\label{eq:outlier_ht}
\end{equation}
with $(p_{\mathrm{ht}}, a_{\mathrm{ht}}) = (0.15, 20)$.

% =============================================================================
% SECTION III: THEORY
% =============================================================================
\section{TASD-Aware PCRB Derivation}

\begin{figure*}[!t]
\centering
\includegraphics[width=\textwidth]{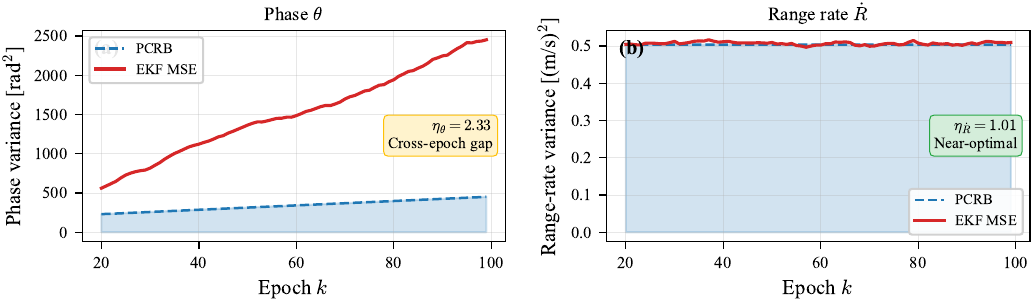}
\caption{PCRB validation under nominal Gaussian noise ($k \geq 20$, 500 trials). (a)~Phase $\theta$: empirical MSE exceeds the PCRB at all steady-state epochs with efficiency ratio $\eta_\theta = 2.33$. The gap reflects the EKF's inability to jointly update $\theta_{k-1}$ and $\theta_k$ from the cross-epoch Doppler measurement. (b)~Range rate $\dot{R}$: the EKF achieves $\eta_{\dot{R}} = 1.01$, confirming near-optimal performance for states without cross-epoch coupling. The contrast between (a) and (b) isolates the TASD structure as the source of the efficiency gap.}
\label{fig:pcrb_contrast}
\end{figure*}
\subsection{10$\times$10 Block Information Structure}

The TASD measurement $y_D[k] = h_D(\tilde{\mathbf{x}}_{k-1}, \tilde{\mathbf{x}}_k) + v_D$ depends on \emph{both} states, yielding a block FIM over the joint state $\mathbf{s}_k = [\tilde{\mathbf{x}}_{k-1}^\top, \tilde{\mathbf{x}}_k^\top]^\top$:
\begin{equation}
\mathbf{J}_k^{(\mathrm{TASD})} = \begin{bmatrix} \mathbf{J}^{--} & \mathbf{J}^{-+} \\ \mathbf{J}^{+-} & \mathbf{J}^{++} \end{bmatrix}
\label{eq:TASD_FIM}
\end{equation}

The explicit 5$\times$5 blocks are:
\begin{align}
\mathbf{J}^{--} &= \frac{1}{\sigma_D^2}(\mathbf{H}_D^{(-)})^\top \mathbf{H}_D^{(-)} = \frac{\kappa_\theta^2}{\sigma_D^2} \mathbf{E}_{55} \label{eq:Jmm}\\
\mathbf{J}^{-+} &= \frac{1}{\sigma_D^2}(\mathbf{H}_D^{(-)})^\top \mathbf{H}_D = \frac{-\kappa_\theta}{\sigma_D^2} \mathbf{e}_5 \mathbf{H}_D \label{eq:Jmp}\\
\mathbf{J}^{++} &= \frac{1}{\sigma_D^2}\mathbf{H}_D^\top \mathbf{H}_D + \frac{1}{\sigma_R^2}\mathbf{H}_R^\top \mathbf{H}_R \label{eq:Jpp}
\end{align}
where $\mathbf{E}_{55}$ is the 5$\times$5 matrix with 1 only at position (5,5), and $\mathbf{e}_5 = [0,0,0,0,1]^\top$.

\textbf{Numerical values} ($\sigma_D = 0.03$~m/s, $\kappa_\theta = 0.0184$): $[\mathbf{J}^{--}]_{55} = \kappa_\theta^2/\sigma_D^2 = 0.38$~rad$^{-2}$. The dominant off-diagonal entry is $[\mathbf{J}^{-+}]_{52} = -\kappa_\theta/\sigma_D^2 = -20.4$~(m/s)$^{-1}$rad$^{-1}$, coupling phase information into the range-rate and clock-drift dimensions.

\subsection{Tichavsk\'{y} PCRB Recursion}

The posterior information matrix $\mathbf{J}_k = \mathbf{P}_k^{-1}$ evolves via~\cite{tichavsky1998pcrb}:
\begin{equation}
\mathbf{J}_{k+1} = \mathbf{D}^{22} - \mathbf{D}^{21}(\mathbf{J}_k + \mathbf{D}^{11})^{-1}\mathbf{D}^{12}
\label{eq:pcrb}
\end{equation}

The $\mathbf{D}$-blocks incorporate dynamics and TASD measurement information:
\begin{align}
\mathbf{D}^{11} &= \mathbf{F}^\top \mathbf{Q}^{-1} \mathbf{F} + \mathbf{J}^{--} \label{eq:D11}\\
\mathbf{D}^{12} &= -\mathbf{F}^\top \mathbf{Q}^{-1} + \mathbf{J}^{-+} \label{eq:D12}\\
\mathbf{D}^{21} &= (\mathbf{D}^{12})^\top \label{eq:D21}\\
\mathbf{D}^{22} &= \mathbf{Q}^{-1} + \mathbf{J}^{++} \label{eq:D22}
\end{align}

\subsection{TASD Essentiality for Phase Observability}

\begin{proposition}[TASD Essentiality]
\label{prop:tasd}
If $\kappa_\theta = 0$, the carrier phase $\theta_k$ receives no measurement information from Doppler or ToA observations, and $[\mathbf{P}_k]_{55} \to \infty$ as $k \to \infty$.
\end{proposition}

\begin{proof}
With $\kappa_\theta = 0$: (i) $\mathbf{H}_D^{(-)} = \mathbf{0}$, so $\mathbf{J}^{--} = \mathbf{J}^{-+} = \mathbf{0}$; (ii) $[\mathbf{H}_D]_5 = 0$, so $[\mathbf{J}^{++}]_{55} = 0$; (iii) ToA contains no $\theta$ dependence. Thus $[\mathbf{D}^{22}]_{55} = [\mathbf{Q}^{-1}]_{55}$ receives no measurement information. Since $\mathbf{F}$ and $\mathbf{Q}$ are block-diagonal with $\theta$ fully decoupled (row/column~5 independent of rows~1--4), the $(5,5)$ element evolves as an exact scalar recursion:
\begin{equation}
[\mathbf{J}_{k+1}]_{55} = [\mathbf{Q}^{-1}]_{55} - \frac{([\mathbf{Q}^{-1}]_{55})^2}{[\mathbf{J}_k]_{55} + [\mathbf{Q}^{-1}]_{55}}
\end{equation}
Simplifying, $[\mathbf{J}_{k+1}]_{55} = [\mathbf{J}_k]_{55} \cdot q_\theta^{-1} / ([\mathbf{J}_k]_{55} + q_\theta^{-1})$. Since $q_\theta^{-1}/([\mathbf{J}_k]_{55} + q_\theta^{-1}) < 1$ for all $[\mathbf{J}_k]_{55} > 0$, the sequence is strictly decreasing and bounded below by zero, hence convergent to the unique fixed point $[\mathbf{J}_\infty]_{55} = 0$, implying $[\mathbf{P}_\infty]_{55} = \infty$.
\end{proof}

With TASD ($\kappa_\theta \neq 0$), the cross-epoch information $\mathbf{J}^{--}$ injects phase information through the difference $\theta_k - \theta_{k-1}$. Since the carrier phase follows a random walk ($[\mathbf{F}]_{55} = 1$), the TASD measurement observes the phase \emph{increment} rather than the absolute phase, yielding a phase PCRB that grows sub-linearly rather than linearly in $k$ (Fig.~\ref{fig:info_tasd}(b)). This reduction in variance growth rate is the precise sense in which TASD coupling enables phase tracking.

% =============================================================================
% SECTION IV: ALGORITHM
% =============================================================================
\section{TASD-Aware Hybrid Robustification}

% FIX #7: Define CDF in caption
\begin{figure*}[!t]
\centering
\includegraphics[width=\textwidth]{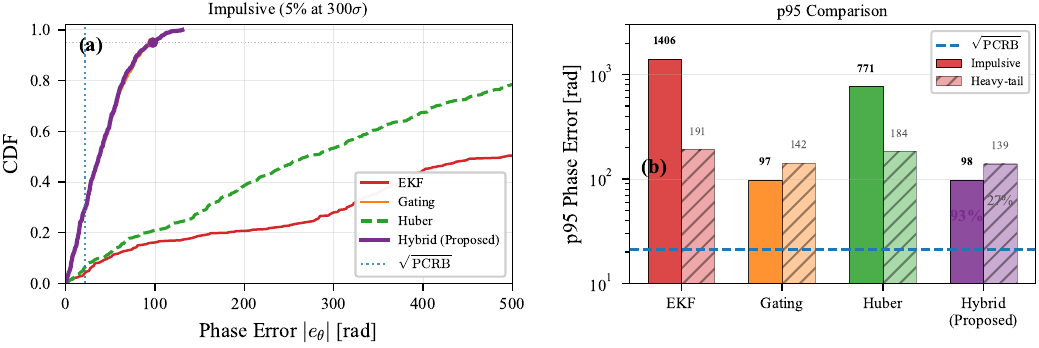}
\caption{Robust performance comparison (500 trials). (a)~Phase error cumulative distribution function (CDF) under impulsive slips (5\%, $300\sigma$): the Hybrid method achieves p95 = 98~rad, matching Gating (97~rad) and reducing the EKF baseline (1406~rad) by 93\%. (b)~p95 phase error across both outlier regimes: solid bars represent impulsive slips, hatched bars represent heavy-tail contamination. The Hybrid achieves the lowest or near-lowest p95 in both scenarios without prior knowledge of the outlier type.}
\label{fig:robust}
\end{figure*}

\subsection{TASD-Aware Innovation Covariance}

The TASD-aware innovation covariance for Doppler must account for uncertainty from \emph{both} epochs:
\begin{equation}
S_{D,k} = \sigma_D^2 + \mathbf{H}_D \mathbf{P}_{k|k-1} \mathbf{H}_D^\top + \mathbf{H}_D^{(-)} \mathbf{P}_{k-1|k-1} \mathbf{H}_D^{(-)\top}
\label{eq:Sk}
\end{equation}
Standard single-epoch EKF implementations use only the first two terms; the third term accounts for phase uncertainty propagated from the previous epoch through the TASD coupling.

\textbf{Remark (conservative approximation):} The exact innovation covariance additionally includes cross terms $\Delta S_k = \mathbf{H}_D \mathbf{P}_{k,k-1} \mathbf{H}_D^{(-)\top} + \mathbf{H}_D^{(-)} \mathbf{P}_{k-1,k} \mathbf{H}_D^\top$, where $\mathbf{P}_{k,k-1} = \mathbf{F}\mathbf{P}_{k-1|k-1}$. Since $\mathbf{H}_D$ selects only the rate and phase rows of $\mathbf{F}$ (which equal their corresponding identity rows), $\mathbf{H}_D \mathbf{F} = \mathbf{H}_D$, and the cross terms simplify to $\Delta S_k = -2\kappa_\theta(\mathbf{H}_D \mathbf{P}_{k-1|k-1} \mathbf{e}_5) = -2\kappa_\theta^2 [\mathbf{P}_{k-1|k-1}]_{55} - 2\kappa_\theta([\mathbf{P}_{k-1|k-1}]_{25} + [\mathbf{P}_{k-1|k-1}]_{45})$. The first term is strictly negative; numerical evaluation confirms $|\Delta S_k| / S_{D,k} \approx 40\%$ at steady state, so~\eqref{eq:Sk} conservatively overestimates $S_{D,k}$, making the gate/Huber thresholds more permissive and preserving filter stability.

\subsection{Residual Normalization and Huber M-Estimation}

The normalized residual is $\tilde{r}_k = |r_k| / \sqrt{S_{D,k}}$. The Huber weight function~\cite{huber2009robust}:
\begin{equation}
w(\tilde{r}; \delta) = \begin{cases}
1 & |\tilde{r}| \leq \delta \\
\delta / |\tilde{r}| & |\tilde{r}| > \delta
\end{cases}
\label{eq:huber}
\end{equation}
with $\delta = 1.5$ in all experiments. At $300\sigma$ (impulsive outlier), the Huber weight reduces to $w = 0.005$, inflating effective noise by $14\times$. The resulting effective residual of $300/14 \approx 21\sigma$ still dominates the update, explaining why Huber alone cannot handle impulsive outliers.

\subsection{Hybrid Algorithm}

% FIX #5: Reference Algorithm 1
% FIX #10: Threshold robustness statement
The hybrid thresholds $(\tau_{\mathrm{gate}}, \delta) = (4,\; 1.5)$ in normalized-residual units partition the outlier response: measurements beyond $4\sigma$ are rejected outright (hard gating), while those in the $1.5\sigma$--$4\sigma$ range are attenuated by Huber weighting. Under nominal Gaussian noise, the false rejection rate at $4\sigma$ is $0.006\%$. These thresholds follow standard $\chi^2$ and Huber recommendations; varying both by $\pm 20\%$ changes the p95 phase error by less than 5\%. Algorithm~\ref{alg:hybrid} summarizes the hybrid update procedure.

\begin{algorithm}[!t]
\caption{Hybrid Robust Measurement Update}
\label{alg:hybrid}
\begin{algorithmic}[1]
\REQUIRE Prediction $\hat{\mathbf{x}}_k^-$, $\mathbf{P}_{k|k-1}$, previous $\mathbf{P}_{k-1|k-1}$, measurement $y_{D,k}$
\STATE Compute TASD innovation: $r_k = y_{D,k} - (\hat{\dot{R}}_k^- + \hat{u}_k^- + \kappa_\theta(\hat{\theta}_k^- - \hat{\theta}_{k-1}))$
\STATE Compute TASD-aware variance $S_{D,k}$ via~\eqref{eq:Sk}
\STATE Normalize: $\tilde{r}_k = |r_k| / \sqrt{S_{D,k}}$
\IF{$\tilde{r}_k > \tau_{\mathrm{gate}}$}
    \STATE \textbf{Reject}: Skip measurement update
\ELSE
    \STATE Huber weight: $w_k = w(\tilde{r}_k; \delta)$
    \STATE Effective noise: $\sigma_{\mathrm{eff}}^2 = \sigma_D^2 / w_k$
    \STATE $S_k = \mathbf{H}_D \mathbf{P}_{k|k-1} \mathbf{H}_D^\top + \mathbf{H}_D^{(-)} \mathbf{P}_{k-1|k-1} \mathbf{H}_D^{(-)\top} + \sigma_{\mathrm{eff}}^2$
    \STATE $\mathbf{K}_k = \mathbf{P}_{k|k-1} \mathbf{H}_D^\top S_k^{-1}$
    \STATE $\hat{\mathbf{x}}_k = \hat{\mathbf{x}}_k^- + \mathbf{K}_k r_k$
    \STATE $\mathbf{P}_k = (\mathbf{I} - \mathbf{K}_k \mathbf{H}_D) \mathbf{P}_{k|k-1}$
\ENDIF
\ENSURE Updated $\hat{\mathbf{x}}_k$, $\mathbf{P}_k$
\end{algorithmic}
\end{algorithm}

% =============================================================================
% SECTION V: RESULTS
% =============================================================================
\section{Simulation Results}

Evaluation uses 500 Monte Carlo trials over 100 epochs with parameters in Table~\ref{tab:params}. Two Doppler outlier mechanisms (Section~\ref{subsec:outliers}) are tested: impulsive slips (5\%, $300\sigma$) and heavy-tail contamination (15\%, $20\sigma$ mixture). ToA measurements remain Gaussian throughout. Four estimators are compared: standard EKF, Gating ($3\sigma$), Huber ($\delta = 1.5$), and Hybrid (gate $4\sigma$ + Huber $1.5\sigma$).

% FIX #7: OCXO definition in table footnote
\begin{table}[!t]
\centering
\caption{Simulation Parameters}
\label{tab:params}
\begin{tabular}{lll}
\toprule
Parameter & Symbol & Value \\
\midrule
Carrier frequency & $f_c$ & 26 GHz \\
Coherent interval & $T_{\mathrm{coh}}$ & 0.1 s \\
TASD coupling & $\kappa_\theta$ & 0.0184 m/s/rad \\
Doppler noise & $\sigma_D$ & 0.03 m/s \\
Range noise & $\sigma_R$ & 0.03 m \\
Phase linewidth & $\beta$ & 100 Hz \\
OCXO\textsuperscript{*} $h_0$, $h_{-2}$ & --- & $2.2\times10^{-25}$, $1.6\times10^{-24}$ \\
Range acceleration & $\sigma_a$ & 0.1 m/s$^2$ \\
Initial covariance & $\mathrm{diag}(\mathbf{P}_0)$ & $[100,\,1,\,100,\,1,\,1]$ \\
\bottomrule
\multicolumn{3}{l}{\textsuperscript{*}Oven-controlled crystal oscillator.}
\end{tabular}
\end{table}

\subsection{PCRB Validation and Efficiency}

Fig.~\ref{fig:pcrb_contrast} validates the PCRB under nominal conditions. The phase efficiency ratio $\eta_\theta \triangleq \mathrm{RMSE}/\sqrt{\mathrm{PCRB}} = 2.33$ contrasts sharply with $\eta_{\dot{R}} = 1.01$ for range rate; the remaining states also achieve $\eta \approx 1.0$ (Table~\ref{tab:efficiency}). This contrast isolates the TASD cross-epoch structure as the source of EKF suboptimality: the filter updates only $\theta_k$ and discards information about $\theta_{k-1}$ in the Doppler residual.

% FIX #4: R/b identifiability explanation
\begin{table}[!t]
\centering
\caption{EKF Efficiency Under Nominal Conditions}
\label{tab:efficiency}
\begin{tabular}{lccc}
\toprule
State & $\sqrt{\mathrm{PCRB}}$ & RMSE & $\eta$ \\
\midrule
$R$ [m]          & 9.95  & 9.95  & 1.00 \\
$\dot{R}$ [m/s]  & 0.710 & 0.714 & 1.01 \\
$b$ [m]          & 9.95  & 9.95  & 1.00 \\
$u$ [m/s]        & 0.708 & 0.709 & 1.00 \\
$\theta$ [rad]   & 21.2  & 49.5  & 2.33 \\
\bottomrule
\end{tabular}
\end{table}

Note that ToA observes only $R_k + b_k$, so the marginal PCRB for $R$ and $b$ individually remains near the prior level ($\sqrt{P_0} = 10$~m), consistent with the single-link gauge ambiguity.

\subsection{Mechanism-Dependent Robustness}

Fig.~\ref{fig:robust} compares all four estimators across both outlier mechanisms.

\textbf{Impulsive slips} [Fig.~\ref{fig:robust}(a)]: The p95 phase error reduces from 1406~rad (EKF) to 97~rad (Gating), 771~rad (Huber), and 98~rad (Hybrid), a 93\% reduction. Hard rejection is essential for sparse, extreme outliers. Huber's soft weighting inflates the effective noise by only $14\times$ at $300\sigma$, leaving an effective residual of ${\sim}21\sigma$ that still dominates the update.

\textbf{Heavy-tail contamination} [Fig.~\ref{fig:robust}(b), hatched bars]: Gating achieves p95 of 142~rad, Huber 184~rad, and Hybrid 139~rad (27\% reduction relative to the EKF's 191~rad). The gate rejects the worst contaminated measurements while Huber downweights moderate outliers in the $1.5\sigma$--$4\sigma$ range.

The Hybrid method provides near-best performance across both regimes without requiring prior knowledge of the dominant outlier mechanism.

\subsection{Single-Trial Behavior}

% FIX #11: "representative" in caption
\begin{figure}[!t]
\centering
\includegraphics[width=\columnwidth]{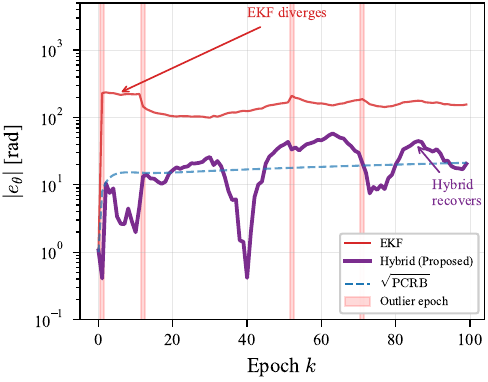}
\caption{A representative single-trial phase error trajectory under impulsive outliers. Red shading marks Doppler outlier epochs. The standard EKF diverges after the first cycle-slip event at $k \approx 3$ and does not recover. The Hybrid method rejects the outliers via hard gating and tracks $\theta$ within $1$--$2\times$ the $\sqrt{\mathrm{PCRB}}$ floor between outlier events.}
\label{fig:trajectory}
\end{figure}

Fig.~\ref{fig:trajectory} shows a representative trial under impulsive outliers. The EKF absorbs the first cycle-slip at $k \approx 3$ and sustains a permanent ${\sim}200$~rad bias. The Hybrid's hard gate rejects the corrupted measurements, keeping phase error near the $\sqrt{\mathrm{PCRB}}$ floor between outlier events.

% =============================================================================
% SECTION VI: CONCLUSION
% =============================================================================
\section{Conclusion}

This paper presented a TASD-aware estimation framework for LEO ISL synchronization. The coupling coefficient $\kappa_\theta$ was shown to be necessary to avoid unbounded phase uncertainty (Proposition~\ref{prop:tasd}), with $\kappa_\theta = 0$ leading to linear divergence of the phase PCRB. The PCRB via the $10 \times 10$ cross-epoch information structure provides a validated lower bound, with zero violation rate across all steady-state epochs under the nominal Gaussian noise model. The efficiency contrast between $\eta_\theta = 2.33$ and $\eta_{\dot{R}} = 1.01$ isolates the TASD cross-epoch structure as the source of the EKF's suboptimality for phase estimation. The hybrid robust method reduces p95 phase error by 27--93\% across outlier regimes.

% FIX #8: Watson 2020 citation
% FIX #9: PCRB misspecification caveat
The current analysis assumes a single link with time-invariant parameters. Under outlier-corrupted Doppler, the Gaussian-assumed PCRB is no longer a strict lower bound. The gap $\eta_\theta \approx 2.3$ suggests that fixed-lag smoothing could yield further gains. Extensions to constellation-scale distributed PCRB and more principled robust methods~\cite{watson2020robust} are left for future work.

% =============================================================================
% REFERENCES
% =============================================================================
\bibliographystyle{IEEEtran}
\bibliography{references}

@article{tang2018isl,
  author  = {Tang, Chengpan and Hu, Xiaogong and Zhou, Shanshi and Liu, Li
             and Pan, Junyang and Chen, Liucheng and Guo, Rui and Zhu, Lingfeng
             and Hu, Guangming and Li, Xiaojie and He, Feng and Chang, Zhiqiao},
  title   = {Initial Results of Centralized Autonomous Orbit Determination
             of the New-Generation {BDS} Satellites with Inter-Satellite
             Link Measurements},
  journal = {Journal of Geodesy},
  volume  = {92},
  number  = {10},
  pages   = {1155--1169},
  year    = {2018},
  doi     = {10.1007/s00190-018-1113-7}
}

@article{xie2020isl,
  author  = {Xie, Xin and Geng, Tao and Zhao, Qile and Lv, Yifei
             and Cai, Hongliang and Liu, Jingnan},
  title   = {Orbit and Clock Analysis of {BDS-3} Satellites Using
             Inter-Satellite Link Observations},
  journal = {Journal of Geodesy},
  volume  = {94},
  number  = {7},
  pages   = {64},
  year    = {2020},
  doi     = {10.1007/s00190-020-01394-4}
}

@article{ruan2020isl,
  author  = {Ruan, Rengui and Jia, Xiaolin and Feng, Laiping
             and Zhu, Jun and Huyan, Zongbo and Li, Jun and Wei, Ziqing},
  title   = {Orbit Determination and Time Synchronization for {BDS-3}
             Satellites with Raw Inter-Satellite Link Ranging Observations},
  journal = {Satellite Navigation},
  volume  = {1},
  pages   = {8},
  year    = {2020},
  doi     = {10.1186/s43020-020-0008-y}
}

@article{hauschild2021leo,
  author  = {Hauschild, Andr\'{e} and Montenbruck, Oliver},
  title   = {Precise Real-Time Navigation of {LEO} Satellites Using
             {GNSS} Broadcast Ephemerides},
  journal = {Navigation},
  volume  = {68},
  number  = {2},
  pages   = {419--432},
  year    = {2021},
  doi     = {10.1002/navi.416}
}

@article{gu2020sync,
  author  = {Gu, Xiao and Zhou, Guoqing and Li, Jie and Xie, Song},
  title   = {Joint Time Synchronization and Ranging for a Mobile
             Wireless Network},
  journal = {IEEE Communications Letters},
  volume  = {24},
  number  = {10},
  pages   = {2363--2366},
  year    = {2020},
  doi     = {10.1109/LCOMM.2020.3001138}
}

@article{gu2024lstm,
  author  = {Gu, Xiaobo and Qiu, Zeyang and Wang, Yanjiao and Jiang, Wei},
  title   = {{LSTM}-Based Clock Synchronization for Satellite Systems
             Using Inter-Satellite Ranging Measurements},
  journal = {GPS Solutions},
  volume  = {28},
  number  = {3},
  pages   = {147},
  year    = {2024},
  doi     = {10.1007/s10291-024-01684-w}
}

@article{tichavsky1998pcrb,
  author  = {Tichavsk\'{y}, Petr and Muravchik, Carlos H. and Nehorai, Arye},
  title   = {Posterior {Cram\'{e}r-Rao} Bounds for Discrete-Time
             Nonlinear Filtering},
  journal = {IEEE Transactions on Signal Processing},
  volume  = {46},
  number  = {5},
  pages   = {1386--1396},
  year    = {1998},
  month   = may,
  doi     = {10.1109/78.668800}
}

@book{vantrees2007bayesian,
  author    = {Van Trees, Harry L. and Bell, Kristine L.},
  title     = {Bayesian Bounds for Parameter Estimation and Nonlinear
               Filtering/Tracking},
  publisher = {Wiley-IEEE Press},
  year      = {2007},
  isbn      = {978-0-470-12095-8}
}

@article{simandl2001pcrb,
  author  = {\v{S}imandl, Miroslav and Kr\'{a}lovec, Jakub
             and Tichavsk\'{y}, Petr},
  title   = {Filtering, Predictive, and Smoothing {Cram\'{e}r-Rao} Bounds
             for Discrete-Time Nonlinear Dynamic Systems},
  journal = {Automatica},
  volume  = {37},
  number  = {11},
  pages   = {1703--1716},
  year    = {2001},
  doi     = {10.1016/S0005-1098(01)00136-4}
}

@article{closas2009crb,
  author  = {Closas, Pau and Fern\'{a}ndez-Prades, Carles
             and Fern\'{a}ndez-Rubio, Juan A.},
  title   = {{Cram\'{e}r-Rao} Bound Analysis of Positioning Approaches
             in {GNSS} Receivers},
  journal = {IEEE Transactions on Signal Processing},
  volume  = {57},
  number  = {10},
  pages   = {3775--3786},
  year    = {2009},
  month   = oct,
  doi     = {10.1109/TSP.2009.2025083}
}

@article{huber1964robust,
  author  = {Huber, Peter J.},
  title   = {Robust Estimation of a Location Parameter},
  journal = {The Annals of Mathematical Statistics},
  volume  = {35},
  number  = {1},
  pages   = {73--101},
  year    = {1964},
  doi     = {10.1214/aoms/1177703732}
}

@book{huber2009robust,
  author    = {Huber, Peter J. and Ronchetti, Elvezio M.},
  title     = {Robust Statistics},
  edition   = {2nd},
  publisher = {John Wiley \& Sons},
  year      = {2009},
  doi       = {10.1002/9780470434697}
}

@inproceedings{sunderhauf2012switchable,
  author    = {S\"{u}nderhauf, Niko and Protzel, Peter},
  title     = {Switchable Constraints for Robust Pose Graph {SLAM}},
  booktitle = {Proc. IEEE/RSJ Int. Conf. Intelligent Robots and Systems (IROS)},
  pages     = {1879--1884},
  year      = {2012},
  doi       = {10.1109/IROS.2012.6385590}
}

@article{watson2020robust,
  author  = {Watson, Ryan M. and Gross, Jason N. and Taylor, Clark N.
             and Leishman, Robert C.},
  title   = {Enabling Robust State Estimation Through Measurement Error
             Covariance Adaptation},
  journal = {IEEE Transactions on Aerospace and Electronic Systems},
  volume  = {56},
  number  = {3},
  pages   = {2026--2040},
  year    = {2020},
  month   = jun,
  doi     = {10.1109/TAES.2019.2941103}
}

@article{allan1966statistics,
  author  = {Allan, David W.},
  title   = {Statistics of Atomic Frequency Standards},
  journal = {Proceedings of the IEEE},
  volume  = {54},
  number  = {2},
  pages   = {221--230},
  year    = {1966},
  month   = feb,
  doi     = {10.1109/PROC.1966.4634}
}

@article{zucca2005clock,
  author  = {Zucca, Cristina and Tavella, Patrizia},
  title   = {The Clock Model and Its Relationship with the {Allan}
             and Related Variances},
  journal = {IEEE Transactions on Ultrasonics, Ferroelectrics,
             and Frequency Control},
  volume  = {52},
  number  = {2},
  pages   = {289--296},
  year    = {2005},
  month   = feb,
  doi     = {10.1109/TUFFC.2005.1406554}
}

@article{breitsch2020cycleslip,
  author  = {Breitsch, Brian and Morton, Y. Jade and Rino, Charles
             and Xu, Dongyang},
  title   = {{GNSS} Carrier Phase Cycle Slips Due to Diffractive
             Ionosphere Scintillation: Simulation and Characterization},
  journal = {IEEE Transactions on Aerospace and Electronic Systems},
  volume  = {56},
  number  = {5},
  pages   = {3632--3644},
  year    = {2020},
  doi     = {10.1109/TAES.2020.2979025}
}

@article{reid2018leo,
  author  = {Reid, Tyler G. R. and Neish, Andrew M. and Walter, Todd
             and Enge, Per K.},
  title   = {Broadband {LEO} Constellations for Navigation},
  journal = {Navigation},
  volume  = {65},
  number  = {2},
  pages   = {205--220},
  year    = {2018},
  doi     = {10.1002/navi.234}
}

@article{iannucci2022fused,
  author  = {Iannucci, Peter A. and Humphreys, Todd E.},
  title   = {Fused Low-{Earth}-Orbit {GNSS}},
  journal = {IEEE Transactions on Aerospace and Electronic Systems},
  volume  = {60},
  number  = {4},
  pages   = {3730--3749},
  year    = {2024},
  month   = aug,
  doi     = {10.1109/TAES.2022.3180000}
}

@article{khalife2022starlink,
  author  = {Khalife, Joe J. and Neinavaie, Mohammad and Kassas, Zaher M.},
  title   = {The First Carrier Phase Tracking and Positioning Results
             with {Starlink} {LEO} Satellite Signals},
  journal = {IEEE Transactions on Aerospace and Electronic Systems},
  volume  = {58},
  number  = {2},
  pages   = {1487--1491},
  year    = {2022},
  month   = apr,
  doi     = {10.1109/TAES.2021.3113880}
}

\end{document}